\documentclass[letterpaper, 10 pt, conference]{ieeeconf}  

\IEEEoverridecommandlockouts                              
\usepackage{bm}
\usepackage{amsmath,amsthm}
\usepackage{graphicx}
\usepackage{amssymb}
\usepackage{epsfig}
\usepackage{float}
\usepackage{mathrsfs}
\usepackage{caption}
\usepackage{subcaption}
\usepackage{longtable}
\usepackage{xtab}
\usepackage{color}
\usepackage{ifthen}
\usepackage{cite}

\usepackage{setspace}

\newtheorem{rem}{Remark}

\newtheorem{thm}{Theorem}

\newtheorem{prop}{Proposition}

\newboolean{showcomments}
\setboolean{showcomments}{true}
\newcommand{\lina}[1]{  \ifthenelse{\boolean{showcomments}}
{ \textcolor{red}{(Lina says:  #1)}} {}  }
\newcommand{\masoud}[1]{  \ifthenelse{\boolean{showcomments}}
{ \textcolor{blue}{(Masoud says:  #1)}} {}  }

\title{\LARGE \bf Distributed Temperature Control via Geothermal Heat Pump Systems in Energy Efficient Buildings}

\author{Xuan Zhang, Wenbo Shi, Qinran Hu, Bin Yan, Ali Malkawi and Na Li 
\thanks{This work was supported by NSF ECCS 1548204 and 1608509, NSF CAREER 1553407, Harvard Center for Green Buildings and Cities, and NARI Group Corporation. X. Zhang, W. Shi, B. Yan and A. Malkawi are with Harvard Center for Green Buildings and Cities, 20 Sumner Road, Cambridge, MA 02138, USA (email: xuan\_zhang@g.harvard.edu, \{wshi, byan, amalkawi\}@gsd.harvard.edu). X. Zhang, Q. Hu and N. Li are with the School of Engineering and Applied Sciences, Harvard University, 29 Oxford Street, Cambridge, MA 02138, USA (email: qinranhu@g.harvard.edu, nali@seas.harvard.edu). \textit{Corresponding Author: X. Zhang.}}
}

\begin{document}

\allowdisplaybreaks

\maketitle
\thispagestyle{empty}
\pagestyle{empty}


\begin{abstract}
Geothermal Heat Pump (GHP) systems are heating and cooling systems that use the ground as the temperature exchange medium. GHP systems are becoming more and more popular in recent years due to their high efficiency. Conventional control schemes of GHP systems are mainly designed for buildings with a single thermal zone. For large buildings with multiple thermal zones, those control schemes either lose efficiency or become costly to implement requiring a lot of real-time measurement, communication and computation. In this paper, we focus on developing energy efficient control schemes for GHP systems in buildings with multiple zones. We present a thermal dynamic model of a building equipped with a GHP system for floor heating/cooling and formulate the GHP system control problem as a resource allocation problem with the objective to maximize user comfort in different zones and to minimize the building energy consumption. We then propose real-time distributed algorithms to solve the control problem. Our distributed multi-zone control algorithms are scalable and do not need to measure or predict any exogenous disturbances such as the outdoor temperature and indoor heat gains. Thus, it is easy to implement them in practice. Simulation results demonstrate the effectiveness of the proposed control schemes.
\end{abstract}

\begin{spacing}{0.96}

\section{Introduction}
According to an investigation by the United Nations, buildings are responsible for $40\%$ of energy consumption, $70\%$ of electricity consumption, and result in $30\%$ of greenhouse gas emission~\cite{UNEP09}. Roughly speaking, Heating Ventilation and Air Conditioning (HVAC) systems in buildings account for $40\%$ of the energy use~\cite{Yan13}. It is therefore necessary to make them more energy efficient for environmental sustainability.

In recent years, Geothermal Heat Pump (GHP) systems are becoming popular among different HVAC systems, due to their highly efficient use of energy, i.e., they can usually deliver more than $3$kWh of heat with $1$kWh of electricity~\cite{TahSR11ifac}. GHP systems are heating/cooling systems that use the ground as the temperature exchange medium. In winter, they transfer heat from the underground soil/water to buildings for heating, and vice versa in summer for cooling. Conventional control of the GHP system includes Proportional-Integral-Derivative (PID) control~\cite{YanPLT07} and centralized Model Predictive Control (MPC)~\cite{HalPMJ12,TahSRM12}. These methods are practically efficient for cases with only a single thermal zone, however, they become either less efficient or costly (due to the centralized operation with heavy burdens of sensing, communication and computation of MPC) for cases with multiple thermal zones for large buildings. Since modern buildings are usually large, complex and are with multiple zones, scalable and easy-implementing control schemes are undoubtedly needed for them if equipped with GHP systems for heating/cooling. 

This paper aims to develop real-time control schemes for GHP systems in typical multi-zone buildings. Specifically, we aim to design distributed algorithms to guide each controllable component to properly adapt their behavior such that system-wide objectives are achieved under given operating conditions. The emergence of distributed/decentralized control in network systems has been stimulated by smart sensing, communication, computing, and actuation technologies nowadays, e.g., in smart grids~\cite{ZhaP15,ZhaKMP15}, smart cities~\cite{ShiLXCG14,ShiLC16}, mobile robots~\cite{JadLM03}, and intelligent transportation systems~\cite{Wan10}. The advantages of distributed/decentralized control include: good scalability as the network grows; reduction of measurement, communication and computation compared with centralized control; privacy preserving. Thus, applying distributed/decentralized control to GHP system control and optimization is becoming an area of active research. Representative work includes, for example, distributed MPC~\cite{ChaMA10,TahSR11ifac}. However, distributed MPC still requires a large amount of sensing, communication and computation. In most cases, it needs good prediction of future disturbances, i.e., outdoor temperature, sunlight, indoor occupancy, etc., which may be hard to obtain in reality. Different from the work using MPC, the controllers designed here are based on solving steady-state optimization problems via gradient algorithms: they (i) are dynamic feedback controllers that can be implemented without measuring or predicting disturbances, (ii) are scalable with respect to building structures, (iii) satisfy the system operating constraints, and (iv) ensure system efficiency, reliability and user comfort.

The structure of this paper is as follows. In Section~\ref{se:setup}, we provide the detailed problem setup, including an introduction of the GHP system, a commonly used thermal dynamic model of the building network, and the optimization problem formulation. Since the original optimization problem is nonconvex, two different scenarios are considered in which the problem can be (approximated and) convexified: (i) the control inputs are only the water flow rates (Section~\ref{se:flowonly}), and (ii) the control inputs are both water flow rates and the heat pump supply temperature (Section~\ref{se:general}). For both scenarios, we use a modified primal-dual gradient method to design real-time distributed/decentralized control schemes. As a result, the thermal dynamics can be driven to equilibria which are the optimal solutions of those associated optimization problems. In Section~\ref{se:simulation}, two numerical examples are provided to illustrate the effectiveness of the designed control schemes, using a building with four adjacent zones. Finally, conclusions and future work are presented in Section~\ref{se:conclusion}.

\noindent\emph{Notation:} The positive projection of a function $h(y)$ on a variable $x\in[0,+\infty)$, $(h(y))_{x}^{+}$ is:
\begin{gather}
(h(y))_{x}^{+}=\left\{ \begin{array}{ll}
h(y) & \textrm{if $x>0$} \\
\max(0,h(y)) & \textrm{if $x=0$}
\end{array} \right.. \nonumber
\end{gather}




\section{Problem Setup}\label{se:setup}

\subsection{GHP system in buildings}
The schematic of a typical GHP system is illustrated in Figure~\ref{fig:1}, which consists of two hydronic and one refrigerant circuits, interconnected by two heat exchangers, i.e., an evaporator and a condenser~\cite{TahSR11ifac,TahSRM12}. In the following we take the heating mode case as an example to explain the working process of the GHP system according to~\cite{Tah12,HalPMJ12}, as the heating mode is more commonly used in practice.

\begin{figure}[!t]
\centering
\includegraphics[width=0.47\textwidth]{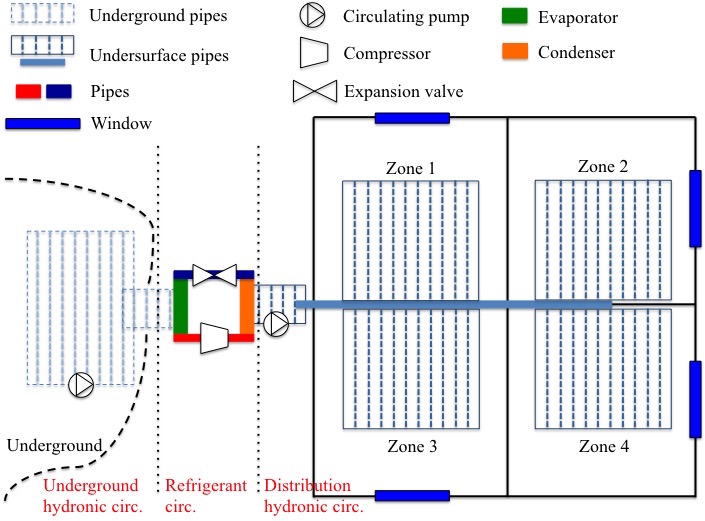}
\caption{Schematic of a typical GHP system.}
\label{fig:1}
\end{figure}

The underground hydronic circuit contains a mixture of water and anti-freeze driven by a small circulating pump, and the temperature of its underground buried brine-filled side is relatively constant with a seasonal pattern, i.e., warmer in winter and cooler in summer than the outdoor temperature. The liquid refrigerant in the refrigerant circuit/the \emph{heat pump}, first goes into the evaporator to absorb heat from the underground hydronic circuit and is converted to its gaseous state. Then this gaseous refrigerant passes through the compressor, and stops at the condenser in which it is converted to its liquid state to heat up the water in the distribution hydronic circuit. Finally this liquid refrigerant passes through the expansion valve, and stops at the evaporator for the next circulation. The distribution hydronic circuit is a grid of indoor under-surface pipes filled with water. Driven by another small circulating pump, this grid of pipes distributes heat to concrete floors (i.e., floor heating) or hydronic radiators (i.e., radiator heating) of a building for heating purpose. \emph{Here we only consider floor heating while radiator heating is similar and will be reported in a future paper}.

In general, the heat pump consumes electrical power to transfer heat to the water in the distribution hydronic circuit. The amount of this heat depends on the flow rate of the water, the heat pump supply/forward temperature, and the return water temperature~\cite{TahSMR11,TahSRM12}. Each zone in the building is equipped with a Thermal Wax Actuator (TWA) that adjusts the valve opening of the pipes for regulating the flow rate. The supply temperature is adjusted by regulating the compressor of the heat pump. The return water temperature can be approximated by the floor temperature which is accurate enough for control design~\cite{Tah12}. These facts will be used later in system modeling as well as control design.

\subsection{Thermal dynamic model with a GHP system}
According to the above configuration, we model a given building as an undirected connected graph $(\mathcal{N},\mathcal{E})$. Here $\mathcal{N}$ is the set of nodes representing zones/rooms, and $\mathcal{E}\subseteq\mathcal{N}\times\mathcal{N}$ is the set of edges. An edge $(i,j)\in\mathcal{E}$ means that zones $i$ and $j$ are neighbors. Let $\mathcal{N}(i)$ denote the set of neighboring zones of zone $i$. The thermal dynamics for each zone is described by a reduced Resistance-Capacitance (RC) model~\cite{LinMB12} (more discussion on this model is available in Remark 1):
\begin{gather}\label{equ:thermalmodel}
C_i\dot T_i=\frac{T^o-T_i}{R_i}+\sum_{j\in\mathcal{N}(i)}\frac{T_j-T_i}{R_{ij}}+\frac{T_{fi}-T_i}{R_{afi}}+Q_i
\end{gather}
where $i\in\mathcal{N}$, $C_i$ is the thermal capacitance, $T_i$ is the indoor temperature, $T^o$ is the outdoor temperature, $R_i$ is the thermal resistance of the wall and window separating zone $i$ and outside, $R_{ij}$ is the thermal resistance of the wall separating zones $i$ and $j$, $T_{fi}$ is the floor temperature, $R_{afi}$ is the thermal resistance between the indoor air and the floor, and $Q_i\ge0$ is the heat gain/disturbance from exogenous sources (e.g., user activity, solar radiation and device operation).

The thermal dynamics of floors equipped with a GHP system is described by a simplified lumped element model~\cite{TahSRM12}:
\begin{subequations}\label{equ:floormodel}
\begin{gather}
C_{fi}\dot T_{fi}=\frac{T_i-T_{fi}}{R_{afi}}+\frac{T_{wi}-T_{fi}}{R_{fwi}}, \text{ }i\in\mathcal{N} \\
C_{wi}\dot T_{wi}=\frac{T_{fi}-T_{wi}}{R_{fwi}}+c_wq_i(T_s-T_{fi}), \text{ }i\in\mathcal{N}
\end{gather}
\end{subequations}
where $C_{fi}$ is the thermal capacitance of the floor, $T_{wi}$ is the temperature of the water in pipes, $R_{fwi}$ is the thermal resistance between the floor and the water, $C_{wi}$ is the thermal capacitance of the water, $c_w$ is the specific heat of the water, $q_i$ is the flow rate of the water, and $T_s$ is the supply temperature of the heat pump. Note that (i) $T_s$ is a common variable of the whole building~\cite{TahSRM12}, and (ii) the term $c_wq_i(T_s-T_{fi})$ stands for the heat transfer from the heat pump to the water in undersurface pipes~\cite{TahSMR11}.

\begin{prop}
When the GHP combined with floor heating/cooling system is off ($q_i=0$),~(\ref{equ:thermalmodel})-(\ref{equ:floormodel}) asymptotically converges to an equilibrium point which is uniquely determined by disturbances $T^o, Q_i$. When the GHP combined with floor heating/cooling system is on, the asymptotic convergence property of~(\ref{equ:thermalmodel})-(\ref{equ:floormodel}) remains and the steady state is uniquely determined by disturbances $T^o, Q_i$ and control inputs $q_i, T_s$.
\end{prop}

The above proposition can be directly derived by rearranging~(\ref{equ:thermalmodel})-(\ref{equ:floormodel}) in state-space representation, and showing that the system matrix is Hurwitz (an alternative way is to construct a quadratic Lyapunov function). So the desiderata is to design $q_i, T_s$ only for periods when the GHP combined with floor heating/cooling system is on, more specifically, is to design the dynamics of $q_i, T_s$ to drive~(\ref{equ:thermalmodel})-(\ref{equ:floormodel}) to some desired state.

\subsection{The optimization problem}
In reality, each zone has a desired temperature which is the set point determined by users. The control objective considered in this paper is to regulate the temperature to be close to the set point in each zone, and to minimize the total energy consumption of the GHP system. More specifically, we consider the following steady-state optimization problem:
\begin{subequations}\label{equ:opt}
\begin{gather}
\hspace{-0.1cm}\min_{Z_i,q_i,T_s, Z_{fi}} \sum_{i\in\mathcal{N}}\left[\frac{1}{2}r_i(Z_i-T_i^{set})^2+s\frac{c_wq_i|T_s-Z_{fi}|}{-aT_s+b}\right] \\
\text{s. t. }\frac{T^o-Z_i}{R_i}+\sum_{j\in\mathcal{N}(i)}\frac{Z_j-Z_i}{R_{ij}}+\frac{Z_{fi}-Z_i}{R_{afi}}+Q_i=0 \\
\frac{Z_i-Z_{fi}}{R_{afi}}+c_wq_i(T_s-Z_{fi})=0 \\
0\le q_i\le q_i^{max} \\
T_s^{min}\le T_s\le T_s^{max}
\end{gather}
\end{subequations}
where $i\in\mathcal{N}$ in~(\ref{equ:opt}b)-(\ref{equ:opt}d), $r_i$ and $s$ are positive weight coefficients, $T_i^{set}$ is the temperature set point, $a,b$ are positive coefficients in the Coefficient of Performance (COP, i.e., an indicator of the relationship between the produced heat and consumed electricity by the heat pump~\cite{TahSMR11}), $[0,q_i^{max}]$ is the range of $q_i$, and $[T_s^{min},T_s^{max}]$ is the range of $T_s$. Note that (i) to avoid confusion between steady-state values and temperature dynamics, we use $Z_i, Z_{fi}$ to denote steady-state temperature values whereas $T_i,T_{fi}$ are temperatures in the dynamic model~(\ref{equ:thermalmodel})-(\ref{equ:floormodel}), (ii) $T^o$ and $Q_i$ are exogenous disturbances, and (iii) we have merged the steady-state equations from~(\ref{equ:floormodel}) to obtain~(\ref{equ:opt}c). We assume that problem~(\ref{equ:opt}) is feasible and satisfies Slater's condition~\cite{BoyV04}. Moreover, we have four important remarks.

\noindent$\bullet$ In the objective function~(\ref{equ:opt}a), the term relating to the total energy/electricity consumption (weighted by $s$) is given by $\sum_{i\in\mathcal{N}}\frac{c_wq_i|T_s-Z_{fi}|}{-aT_s+b}$: the term $\sum_{i\in\mathcal{N}}c_wq_i|T_s-Z_{fi}|$ stands for the total heat exchange between the heat pump and the water in pipes; $-aT_s+b>0$ is the COP which has been approximated as a linear function of $T_s$ as in~\cite{TahSMR11}; $a,b$ can usually be obtained from the heat pump data sheet~\cite{TahSMR11,TahSRM12}.

\noindent$\bullet$ Parameters $r_i, s$ are determined by users. If users prefer more comfort, they can increase $r_i$ and decrease $s$, and vice versa. Because of this flexibility, we do not impose constraints on the temperature comfortable range.

\noindent$\bullet$ In the heating mode, $T_s>Z_{fi} (\text{or }T_{fi}), \forall i$ hold; in the cooling mode, $T_s<Z_{fi} (\text{or }T_{fi}), \forall i$ hold. This is usually true in practice. For example~\cite{TahSMR11,TahSRM12}, in the heating mode, $T_s>27^{\circ}C$ while $T_{fi}<26^{\circ}C$. \emph{Once the mode is determined, the sign of $T_s-Z_{fi} (\text{or }T_s-T_{fi})$ is determined.}

\noindent$\bullet$ The inequality constraints~(\ref{equ:opt}d)-(\ref{equ:opt}e) are in accord with those in the optimization problem (7) in~\cite{TahSRM12}.

To conclude, the goal is to design the regulating rule for $q_i, T_s$ so that system~(\ref{equ:thermalmodel})-(\ref{equ:floormodel}) can be driven to an equilibrium point which is the optimal solution to problem~(\ref{equ:opt}).

\begin{rem}{\hspace{-0.1cm}\cite{ZhaSL17}}
Though system~(\ref{equ:thermalmodel}) is a 1st-order RC model, using higher order RC models does not affect the formulation of~(\ref{equ:opt}) since it is a steady-state optimization problem. For example, for the 2nd-order model in~\cite{HaoLKS15,LinMB12}
\begin{gather}
C_i\dot T_i=\frac{T^o-T_i}{R_i}+\sum_{j\in\mathcal{N}(i)}\frac{T_{ij}-T_i}{R_{ij}}+\frac{T_{fi}-T_i}{R_{afi}}+Q_i \nonumber\\
C_{ij}\dot T_{ij}=\frac{T_i-T_{ij}}{R_{ij}}+\frac{T_j-T_{ij}}{R_{ij}} \nonumber
\end{gather}
where $T_{ij}$ is the temperature of the wall separating zones $i$ and $j$, and $C_{ij}$ is the thermal capacitance of the wall, the corresponding steady-state Equation~(\ref{equ:opt}b) is given by
\begin{gather}
\frac{T^o-Z_i}{R_i}+\sum_{j\in\mathcal{N}(i)}\frac{Z_j-Z_i}{2R_{ij}}+\frac{Z_{fi}-Z_i}{R_{afi}}+Q_i=0 \nonumber
\end{gather}
which results from the steady-state equation $Z_{ij}=\frac{Z_i+Z_j}{2}$ ($Z_{ij}$ is the steady-state temperature of the wall), i.e., the steady-state equation of the higher order model can be reduced to~(\ref{equ:opt}b) by eliminating states of solids in the building envelope. Since our control design procedures proposed later are based on solving the steady-state optimization problem~(\ref{equ:opt}), using higher order models will not affect them.
\end{rem}

\section{Scenario I: Flow Rate Control Only}\label{se:flowonly}

\subsection{Problem reformulation}
In this section, we consider the water flow rate $q_i$ as the only control input to each zone, and regard $T_s$ as a known exogenous signal which satisfies constraint~(\ref{equ:opt}e). Such scenario could happen in practice, for instance, $T_s$ is required to track some prescribed curve~\cite{TahSR11ifac}. In this case, problem~(\ref{equ:opt}) can be simplified into
\begin{subequations}\label{equ:optnoTs}
\begin{gather}
\min_{Z_i,u_i, Z_{fi}} \sum_{i\in\mathcal{N}}\left[\frac{1}{2}r_i(Z_i-T_i^{set})^2+s\frac{|u_i|}{-aT_s+b}\right] \\
\text{s. t. }\frac{T^o-Z_i}{R_i}+\sum_{j\in\mathcal{N}(i)}\frac{Z_j-Z_i}{R_{ij}}+\frac{Z_{fi}-Z_i}{R_{afi}}+Q_i=0 \\
\frac{Z_i-Z_{fi}}{R_{afi}}+u_i=0 \\
0\le \frac{u_i}{c_w(T_s-Z_{fi})}\le q_i^{max}
\end{gather}
\end{subequations}
where $i\in\mathcal{N}$ in~(\ref{equ:optnoTs}b)-(\ref{equ:optnoTs}d), $u_i=c_wq_i(T_s-Z_{fi})$ is introduced to replace terms on $q_i$, and constraint~(\ref{equ:opt}e) is dropped (under these actions, problems~(\ref{equ:opt}) and~(\ref{equ:optnoTs}) are still equivalent). Note that the GHP system is in either the heating mode or the cooling mode. Once the mode is determined, the signs of $u_i$ and $T_s-Z_{fi}$ are determined so that (i) $|u_i|$ equals either $u_i$ or $-u_i$, and (ii) the inequality constraint~(\ref{equ:optnoTs}d) can become linear by multiplying $c_w(T_s-Z_{fi})$ on both sides. Thus, problem~(\ref{equ:optnoTs}) naturally becomes convex.

\subsection{A distributed algorithm}
Once the mode is determined, problem~(\ref{equ:optnoTs}) can be solved in either a centralized or distributed/decentralized way. Any centralized algorithm requires to measure the outdoor temperature $T^o$ and the indoor heat gain $Q_i$ in every zone ($T_s$ is given). Because these exogenous disturbances can fluctuate frequently and are not easy to obtain, the cost of centralized algorithms would be expensive. Next we develop a real-time distributed algorithm that does not need measurement of these exogenous disturbances.

\emph{Consider the heating mode case} (the cooling mode case is similar). The Lagrangian function of~(\ref{equ:optnoTs}) is given by
\begin{align}
L=&\sum_{i\in\mathcal{N}}\left[\frac{1}{2}r_i(Z_i-T_i^{set})^2+s\frac{u_i}{-aT_s+b}\right] \nonumber\\
&+\sum_{i\in\mathcal{N}}\zeta_i\Big(\frac{T^o-Z_i}{R_i}+\sum_{j\in\mathcal{N}(i)}\frac{Z_j-Z_i}{R_{ij}}+\frac{Z_{fi}-Z_i}{R_{afi}}+Q_i\Big) \nonumber\\
&+\sum_{i\in\mathcal{N}}\lambda_i\Big(\frac{Z_i-Z_{fi}}{R_{afi}}+u_i\Big)+\sum_{i\in\mathcal{N}}\mu_i^-(-u_i) \nonumber\\
&+\sum_{i\in\mathcal{N}}\mu_i^+(u_i-q_i^{max}c_w(T_s-Z_{fi})) \nonumber
\end{align}
where $\zeta_i, \lambda_i, \mu_i^+, \mu_i^-$ are the Lagrange multipliers/dual variables for constraints~(\ref{equ:optnoTs}b)-(\ref{equ:optnoTs}d). Since problem~(\ref{equ:optnoTs}) is convex, feasible and satisfies Slater's condition, the Karush-Kuhn-Tucker (KKT) conditions are necessary and sufficient conditions for optimality~\cite{BoyV04}, given by
\begin{subequations}\label{equ:kkt}
\begin{align}
\frac{\partial L}{\partial Z_i}=&r_i(Z_i-T_i^{set})-\zeta_i\Big(\frac{1}{R_i}+\sum_{j\in\mathcal{N}(i)}\frac{1}{R_{ij}}+\frac{1}{R_{afi}}\Big) \nonumber\\
&+\sum_{j\in\mathcal{N}(i)}\frac{\zeta_j}{R_{ij}}+\frac{\lambda_i}{R_{afi}}=0, \text{ }i\in\mathcal{N} \\
\frac{\partial L}{\partial u_i}=&\frac{s}{-aT_s+b}+\lambda_i+\mu_i^+-\mu_i^-=0, \text{ }i\in\mathcal{N} \\
\frac{\partial L}{\partial Z_{fi}}=&\frac{\zeta_i-\lambda_i}{R_{afi}}+\mu_i^+q_i^{max}c_w=0, \text{ }i\in\mathcal{N} \\
\frac{\partial L}{\partial\zeta_i}=&\Big(\frac{T^o-Z_i}{R_i}+\sum_{j\in\mathcal{N}(i)}\frac{Z_j-Z_i}{R_{ij}}+\frac{Z_{fi}-Z_i}{R_{afi}}+Q_i\Big) \nonumber\\
=&0, \text{ }i\in\mathcal{N} \\
\frac{\partial L}{\partial\lambda_i}=&\Big(\frac{Z_i-Z_{fi}}{R_{afi}}+u_i\Big)=0, \text{ }i\in\mathcal{N} \\
\mu_i^+&(u_i-q_i^{max}c_w(T_s-Z_{fi}))=0 \nonumber\\
\mu_i^+&\ge0, u_i-q_i^{max}c_w(T_s-Z_{fi})\le0, \text{ }i\in\mathcal{N} \\
\mu_i^-&(-u_i)=0, \text{ }\mu_i^-\ge0, -u_i\le0, \text{ }i\in\mathcal{N}.
\end{align}
\end{subequations}

Motivated by a modified primal-dual gradient method~\cite{ZhaP15tac,ZhaLP15cdc}, we design the following algorithm to solve~(\ref{equ:optnoTs}):
\begin{subequations}\label{equ:controller}
\begin{align}
\dot Z_i=&-k_{Z_i}\Big(\frac{\partial L}{\partial Z_i}\Big)=k_{Z_i}\Big(r_i(T_i^{set}-Z_i)+\zeta_i\Big(\frac{1}{R_i}+\frac{1}{R_{afi}} \nonumber\\
&+\sum_{j\in\mathcal{N}(i)}\frac{1}{R_{ij}}\Big)-\sum_{j\in\mathcal{N}(i)}\frac{\zeta_j}{R_{ij}}-\frac{\lambda_i}{R_{afi}} \Big) \\
\dot u_i=&-k_{u_i}\Big(\frac{\partial L}{\partial u_i}+k_{eu_i}(u_i-\hat u_i)\Big)=k_{u_i}\Big(\frac{s}{aT_s-b}-\lambda_i \nonumber\\
&-\mu_i^++\mu_i^-+k_{eu_i}(\hat u_i-u_i) \Big) \\
\dot{\hat u}_i=&\hat k_{eu_i}(u_i-\hat u_i) \\
\dot Z_{fi}=&-k_{Z_{fi}}\Big(\frac{\partial L}{\partial Z_{fi}}+k_{eZ_{fi}}(Z_{fi}-\hat Z_{fi})\Big)=k_{Z_{fi}}\Big(\frac{\lambda_i-\zeta_i}{R_{afi}} \nonumber\\
&-\mu_i^+q_i^{max}c_w+k_{eZ_{fi}}(\hat Z_{fi}-Z_{fi}) \Big) \\
\dot{\hat Z}_{fi}=&\hat k_{eZ_{fi}}(Z_{fi}-\hat Z_{fi}) \\
\dot\zeta_i=&k_{\zeta_i}\Big(\frac{\partial L}{\partial \zeta_i}\Big)=k_{\zeta_i}\Big(\frac{T^o-Z_i}{R_i}+\sum_{j\in\mathcal{N}(i)}\frac{Z_j-Z_i}{R_{ij}} \nonumber\\
&+\frac{Z_{fi}-Z_i}{R_{afi}}+Q_i\Big) \\
\dot\lambda_i=&k_{\lambda_i}\Big(\frac{\partial L}{\partial \lambda_i}\Big)=k_{\lambda_i}\Big(\frac{Z_i-Z_{fi}}{R_{afi}}+u_i\Big) \\
\dot\mu_i^+=&k_{\mu_i^+}\Big(\frac{\partial L}{\partial \mu_i^+}\Big)_{\mu_i^+}^+=k_{\mu_i^+}(u_i-q_i^{max}c_w(T_s-Z_{fi}))_{\mu_i^+}^+ \\
\dot\mu_i^-=&k_{\mu_i^-}\Big(\frac{\partial L}{\partial \mu_i^-}\Big)_{\mu_i^-}^+=k_{\mu_i^-}(-u_i)_{\mu_i^-}^+
\end{align}
\end{subequations}
where $i\in\mathcal{N}$, $k_{Z_i}, k_{u_i}, k_{eu_i}, \hat k_{eu_i}, k_{Z_{fi}}, k_{eZ_{fi}}, \hat k_{eZ_{fi}}, k_{\zeta_i}, k_{\lambda_i},$ $k_{\mu_i^+}, k_{\mu_i^-}$ are positive scalars representing the controller gains, and we have introduced the auxiliary states $\hat u_i, \hat Z_{fi}$: since the objective function~(\ref{equ:optnoTs}a) is not strictly convex in $u_i, Z_{fi}$, a standard primal-dual gradient method~\cite{FeiP10} could yield large oscillations; after introducing the extra dynamics, the transient behavior of the overall system can be improved (demonstrated in Section~\ref{se:simulation}). Note that $T_i, T_{fi}$ have their own dynamics given by~(\ref{equ:thermalmodel})-(\ref{equ:floormodel}) and thus can not be designed, which is why we replace $T_i, T_{fi}$ with $Z_i, Z_{fi}$ initially, i.e., $Z_i, Z_{fi}, i\in\mathcal{N}$ are ancillary state variables. According to~\cite{ZhaLP15cdc,CheMC16}, it is true that~(\ref{equ:controller}) asymptotically converges to an equilibrium point which is the optimal solution of~(\ref{equ:optnoTs}), since the optimization problem is convex and extra dynamics have been included in~(\ref{equ:controller}). Now using
\begin{gather}\label{equ:controllerqi}
q_i=\frac{u_i}{c_w(T_s-Z_{fi})}, \text{ }i\in\mathcal{N}
\end{gather}
as the control input to system~(\ref{equ:thermalmodel})-(\ref{equ:floormodel}), we can naturally obtain a real-time distributed controller to regulate~(\ref{equ:thermalmodel})-(\ref{equ:floormodel}) to a steady state which is the optimal solution to problem~(\ref{equ:optnoTs}) ((\ref{equ:optnoTs}) is equivalent to~(\ref{equ:opt}) under known $T_s$).

\begin{thm}\label{thm:1}
Given constant/step change/slow-varying $T^o,$ $Q_i, T_s$, the trajectory of system~(\ref{equ:thermalmodel})-(\ref{equ:floormodel}) and~(\ref{equ:controller})-(\ref{equ:controllerqi}) asymptotically converges to an equilibrium point at which $T_i, q_i, T_{fi}$ of the equilibrium point is the optimal solution of~(\ref{equ:opt}).
\end{thm}
\begin{proof}
According to~\cite{ZhaLP15cdc,CheMC16}, each trajectory of system~(\ref{equ:controller}) asymptotically converges to an equilibrium point which is the optimal solution of~(\ref{equ:optnoTs}). Under~(\ref{equ:controllerqi}), the resulting equilibrium point after state transformation from $u_i$ to $q_i$, is the optimal solution of~(\ref{equ:opt}). On the other hand, the trajectory of the overall system~(\ref{equ:thermalmodel})-(\ref{equ:floormodel}) and~(\ref{equ:controller})-(\ref{equ:controllerqi}) also asymptotically converges to an equilibrium point, due to the cascade nature, i.e.,~(\ref{equ:controller})-(\ref{equ:controllerqi})$\rightarrow$(\ref{equ:thermalmodel})-(\ref{equ:floormodel}). By Proposition 1, since the equilibrium point of~(\ref{equ:thermalmodel})-(\ref{equ:floormodel}) is uniquely determined by the inputs $T^o, Q_i, q_i, T_s$ in which $q_i$ is given by~(\ref{equ:controllerqi}), we have $T_i=Z_i, T_{fi}=Z_{fi}$ (here $T_i, T_{fi}$ are states given by~(\ref{equ:thermalmodel})-(\ref{equ:floormodel})) when the overall system reaches steady state, which completes the proof.
\end{proof}

This theorem requires $T^o, Q_i$ to be either constant, step change, or slow-varying, which holds in practice as they vary at a time-scale of minutes. Remark that our controller operates in real-time, i.e., at a time-scale of seconds.

In Equation~(\ref{equ:controller}f), the disturbances $T^o, Q_i$ appear. Motivated by~\cite{ZhaLP15cdc}, to make the algorithm implementable without measuring these terms, we introduce $\tilde\zeta_i=\frac{\zeta_i}{k_{\zeta_i}}-C_iT_i$ as
\begin{align}\label{equ:controllerplus}
\dot{\tilde\zeta}_i=&\frac{T_i-Z_i}{R_i}+\sum_{j\in\mathcal{N}(i)}\frac{T_i-Z_i-T_j+Z_j}{R_{ij}} \nonumber\\
&+\frac{T_i-Z_i-T_{fi}+Z_{fi}}{R_{afi}}, \text{ }i\in\mathcal{N}.
\end{align}
Moreover, we substitute $\zeta_i=k_{\zeta_i}(\tilde\zeta_i+C_iT_i)$ into~(\ref{equ:controller}a),~(\ref{equ:controller}d) to eliminate $\zeta_i$. Now the proposed control scheme~(\ref{equ:controller}a)-(\ref{equ:controller}e),~(\ref{equ:controller}g)-(\ref{equ:controller}i) and~(\ref{equ:controllerqi})-(\ref{equ:controllerplus}) is completely distributed as shown in Figure~\ref{fig:2} and can be implemented as follow. Given $T_s, C_i, R_i, R_{ij}, R_{afi}, r_i, s, a, b, q_i^{max}$, each zone in the building collects $T_i^{set}$ from users, locally measures its indoor temperature $T_i$ and floor temperature $T_{fi}$, receives the feedback signals $\zeta_j=k_{\zeta_j}(\tilde\zeta_j+C_jT_j)$ and $T_j-Z_j$ from its neighboring zones, and then uses the information to update $Z_i, u_i, \hat u_i, Z_{fi}, \hat Z_{fi}, \tilde\zeta_i,\lambda_i, \mu_i^+, \mu_i^-, q_i$. Here $C_i, R_i, R_{ij}, R_{afi}, a, b, q_i^{max}$ are building parameters, $r_i, s$ are parameters specified by users, and $T_s$ is a known signal.

\begin{figure}[!t]
\centering
\includegraphics[width=0.47\textwidth]{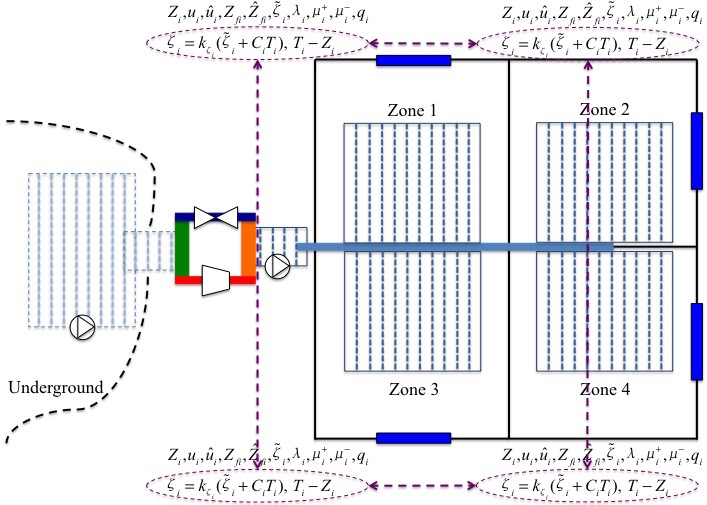}
\caption{Information exchange of the distributed controller.}
\label{fig:2}
\end{figure}

\begin{rem}
In reality, due to that $Z_i, Z_j$ (or $T_i, T_j$) of neighboring zones are often very close to each other and $R_{ij}$ is not small, the total heat gain/loss from neighboring zones is (sometimes much) less dominant compared with the heat gain/loss from the outside plus the indoor heat gain in every zone. Thus, the term $\sum_{j\in\mathcal{N}(i)}\frac{T_j-T_i}{R_{ij}}$ could be ignored in~(\ref{equ:thermalmodel}) as well as the term $\sum_{j\in\mathcal{N}(i)}\frac{Z_j-Z_i}{R_{ij}}$ in~(\ref{equ:opt})/(\ref{equ:optnoTs}). In this case, following the same design procedure, we end up with a completely decentralized control scheme (given by setting $R_{ij}=\infty$ in~(\ref{equ:controller}a) and~(\ref{equ:controllerplus})) that only needs local measurement.
\end{rem}

\section{Scenario II: The General Case}\label{se:general}

\subsection{Problem reformulation}
In this section, both $q_i$ and $T_s$ are considered as control inputs to the system. Instead of handling the nonconvex problem~(\ref{equ:opt}) directly, we focus on an approximate version of this problem, i.e., rather than minimizing the exact energy consumption in the objective function, we minimize one of its upper bound. This approximation can help convexify~(\ref{equ:opt}), which will become clear later. The approximate problem is
\begin{subequations}\label{equ:optapp}
\begin{gather}
\hspace{-0.4cm}\min_{Z_i,q_i,T_s, Z_{fi}} \sum_{i\in\mathcal{N}}\left[\frac{1}{2}r_i(Z_i-T_i^{set})^2+s\frac{c_w^2q_i^2(T_s-Z_{fi})^2}{-aT_s+b}\right] \\
\text{s. t. ~(\ref{equ:opt}b)-(\ref{equ:opt}e)}
\end{gather}
\end{subequations}
where we have modified the term $\sum_{i\in\mathcal{N}}\frac{c_wq_i|T_s-Z_{fi}|}{-aT_s+b}$ in~(\ref{equ:opt}a) to be $\sum_{i\in\mathcal{N}}\frac{c_w^2q_i^2(T_s-Z_{fi})^2}{-aT_s+b}$ in~(\ref{equ:optapp}a). The latter term is actually an upper bound of the square of the exact energy consumption: using a Cauchy-Schwarz inequality, we have
\begin{align}
\sum_{i\in\mathcal{N}}\frac{c_w^2q_i^2(T_s-Z_{fi})^2}{-aT_s+b}\ge&\frac{(\sum_{i\in\mathcal{N}}c_wq_i|T_s-Z_{fi}|)^2}{|\mathcal{N}|(-aT_s+b)} \nonumber\\
\ge&\frac{b-aT_s}{|\mathcal{N}|}\frac{(\sum_{i\in\mathcal{N}}c_wq_i|T_s-Z_{fi}|)^2}{(-aT_s+b)^2} \nonumber\\
\ge&\frac{b-aT_s^{max}}{|\mathcal{N}|}\frac{(\sum_{i\in\mathcal{N}}c_wq_i|T_s-Z_{fi}|)^2}{(-aT_s+b)^2} \nonumber
\end{align}
where $|\mathcal{N}|$ is the number of zones in the building. Therefore, rather than minimizing the total consumption directly, the objective here aims to minimize its upper bound, which is sufficient for energy saving purpose.

Now we show how~(\ref{equ:optapp}) can be turned convex. Again, we introduce variables $u_i=c_wq_i(T_s-Z_{fi}), i\in\mathcal{N}$ to get
\begin{subequations}\label{equ:optappre}
\begin{gather}
\min_{Z_i, u_i, T_s, Z_{fi}} \sum_{i\in\mathcal{N}}\left[\frac{1}{2}r_i(Z_i-T_i^{set})^2+s\frac{u_i^2}{-aT_s+b}\right] \\
\text{s. t. ~(\ref{equ:optnoTs}b)-(\ref{equ:optnoTs}d) and~(\ref{equ:opt}e).}
\end{gather}
\end{subequations}
Note that (i) the function $E=\frac{\sum_{i\in\mathcal{N}}u_i^2}{-aT_s+b}$ is convex in $u_i, T_s$ as its Hessian matrix equals
\begin{align}
\left[ {\begin{array}{*{20}{c}}
{\frac{2}{{b - a{T_s}}}}&0&0& \vdots \\
0& \ddots &0&{\frac{{2a{u_i}}}{{{{(b - a{T_s})}^2}}}}\\
0&0&{\frac{2}{{b - a{T_s}}}}& \vdots \\
 \cdots &{\frac{{2a{u_i}}}{{{{(b - a{T_s})}^2}}}}& \cdots &{\frac{{2{a^2}\sum_{i\in\mathcal{N}} {u_i^2} }}{{{{(b - a{T_s})}^3}}}}
\end{array}} \right] \nonumber
\end{align}
which is positive semi-definite and (ii) once the mode is determined, the signs of $u_i$ and $T_s-Z_{fi}$ are determined so that constraint~(\ref{equ:optnoTs}d) can become linear by multiplying $c_w(T_s-Z_{fi})$ on both sides. Next we design a distributed algorithm to solve the convex optimization problem~(\ref{equ:optappre}).

\subsection{A distributed algorithm}
Similar to Section~\ref{se:flowonly}-B, \emph{we consider the heating mode case}. The design methodology is motivated by a modified primal-dual gradient method~\cite{ZhaP15tac,ZhaLP15cdc}. For simplicity, we directly write down the resulting distributed algorithm:
\begin{subequations}\label{equ:controllerapp}
\begin{align}
\dot Z_i=&k_{Z_i}\Big(r_i(T_i^{set}-Z_i)+\zeta_i\Big(\frac{1}{R_i}+\frac{1}{R_{afi}}+\sum_{j\in\mathcal{N}(i)}\frac{1}{R_{ij}}\Big) \nonumber\\
&-\sum_{j\in\mathcal{N}(i)}\frac{\zeta_j}{R_{ij}}-\frac{\lambda_i}{R_{afi}} \Big) \\
\dot u_i=&k_{u_i}\Big(\frac{2su_i}{aT_s-b}-\lambda_i-\mu_i^++\mu_i^- \Big) \\
\dot T_s=&k_{T_s}\Big(-\frac{as\sum_{i\in\mathcal{N}}u_i^2}{(b-aT_s)^2}+\sum_{i\in\mathcal{N}}\mu_i^+q_i^{max}c_w-\nu^++\nu^- \Big) \\
\dot Z_{fi}=&k_{Z_{fi}}\Big(\frac{\lambda_i-\zeta_i}{R_{afi}}-\mu_i^+q_i^{max}c_w+k_{eZ_{fi}}(\hat Z_{fi}-Z_{fi}) \Big) \\
\dot{\hat Z}_{fi}=&\hat k_{eZ_{fi}}(Z_{fi}-\hat Z_{fi}) \\
\dot\zeta_i=&k_{\zeta_i}\Big(\frac{T^o-Z_i}{R_i}+\sum_{j\in\mathcal{N}(i)}\frac{Z_j-Z_i}{R_{ij}}+\frac{Z_{fi}-Z_i}{R_{afi}}+Q_i\Big) \\
\dot\lambda_i=&k_{\lambda_i}\Big(\frac{Z_i-Z_{fi}}{R_{afi}}+u_i\Big) \\
\dot\mu_i^+=&k_{\mu_i^+}(u_i-q_i^{max}c_w(T_s-Z_{fi}))_{\mu_i^+}^+ \\
\dot\mu_i^-=&k_{\mu_i^-}(-u_i)_{\mu_i^-}^+ \\
\dot\nu^+=&k_{\nu^+}(T_s-T_s^{max})_{\nu^+}^+ \\
\dot\nu^-=&k_{\nu^-}(T_s^{min}-T_s)_{\nu^-}^+
\end{align}
\end{subequations}
where $i\in\mathcal{N}$, $k_{Z_i}, k_{u_i}, k_{T_s}, k_{Z_{fi}}, k_{eZ_{fi}}, \hat k_{eZ_{fi}}, k_{\zeta_i}, k_{\lambda_i}, k_{\mu_i^+},$ $k_{\mu_i^-}, k_{\nu^+}, k_{\nu^-}$ are positive scalars representing the controller gains, and we have introduced the auxiliary state $\hat Z_{fi}$ to improve the performance of the algorithm (since $u_i$ and $T_s$ are coupled in $E$ whose Hessian matrix is not always zero, only adding $\hat Z_{fi}$ is enough for behavior enhancement). Now using~(\ref{equ:controllerqi}) and~(\ref{equ:controllerapp}c) as the control input to~(\ref{equ:thermalmodel})-(\ref{equ:floormodel}) where $u_i, T_s, Z_{fi}$ are determined by~(\ref{equ:controllerapp}), we can naturally obtain a real-time distributed controller to regulate~(\ref{equ:thermalmodel})-(\ref{equ:floormodel}) to a steady state which is the optimal solution to~(\ref{equ:optapp}).

\begin{thm}\label{thm:2}
Given constant/step change/slow-varying $T^o,$ $Q_i$ (remark that they vary at a time-scale of minutes), each trajectory of the overall system~(\ref{equ:thermalmodel})-(\ref{equ:floormodel}),~(\ref{equ:controllerqi}) and~(\ref{equ:controllerapp}) asymptotically converges to an equilibrium point at which $T_i, q_i,T_s, T_{fi}$ of this point is the optimal solution of~(\ref{equ:optapp}).
\end{thm}
\begin{proof}
The proof is similar to that of Theorem 1, and thus, is omitted for brevity.
\end{proof}

In Equation~(\ref{equ:controllerapp}f), the disturbances $T^o, Q_i$ appear. Similar to Section~\ref{se:flowonly}-B, to make the algorithm implementable without measuring these terms, we introduce $\tilde\zeta_i=\frac{\zeta_i}{k_{\zeta_i}}-C_iT_i$ whose dynamics is given by~(\ref{equ:controllerplus}). Moreover, we substitute $\zeta_i=k_{\zeta_i}(\tilde\zeta_i+C_iT_i)$ into~(\ref{equ:controllerapp}a),~(\ref{equ:controllerapp}d) to eliminate $\zeta_i$. Now the proposed control algorithm~(\ref{equ:controllerapp}a)-(\ref{equ:controllerapp}e),~(\ref{equ:controllerapp}g)-(\ref{equ:controllerapp}k) and~(\ref{equ:controllerqi})-(\ref{equ:controllerplus}) is completely distributed as shown in Figure~\ref{fig:3} and can be implemented as follow. Given $C_i, R_i, R_{ij}, R_{afi}, r_i, s, a, b, q_i^{max}$, each zone in the building collects $T_i^{set}$ from users, locally measures its indoor temperature $T_i$ and floor temperature $T_{fi}$, receives the feedback signals $\zeta_j=k_{\zeta_j}(\tilde\zeta_j+C_jT_j), T_j-Z_j$ from its neighboring zones and $T_s$ from the compressor, and then uses the information to update $Z_i, u_i, Z_{fi}, \hat Z_{fi}, \tilde\zeta_i,\lambda_i, \mu_i^+, \mu_i^-, q_i$. On the other hand, given $T_s^{min}, T_s^{max}$, the compressor receives the feedback signals $u_i, \mu_i^+q_i^{max}c_w$ from each zone, updates $T_s, \nu^+, \nu^-$, and then broadcasts $T_s$. Here $C_i, R_i, R_{ij}, R_{afi}, a, b, q_i^{max}, T_s^{min}, T_s^{max}$ are building parameters, and $r_i, s$ are specified by users.

\begin{figure}[!t]
\centering
\includegraphics[width=0.47\textwidth]{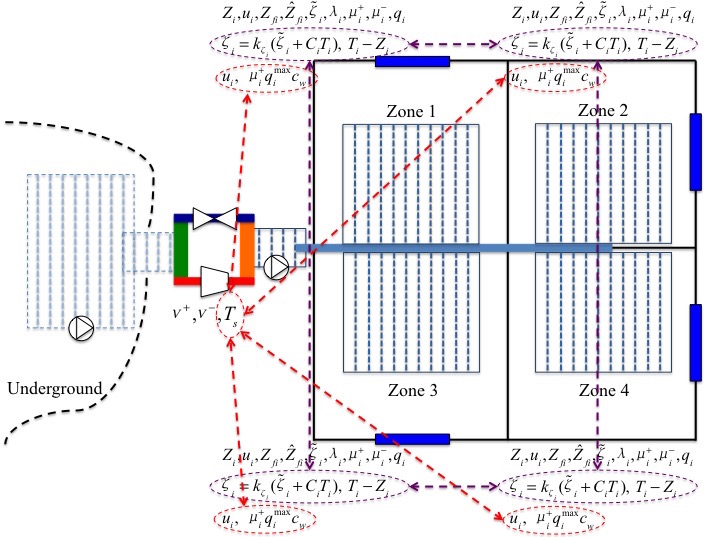}
\caption{Information exchange of the distributed controller.}
\label{fig:3}
\end{figure}

\begin{rem}
Similar to Remark 2, the term $\sum_{j\in\mathcal{N}(i)}\frac{T_j-T_i}{R_{ij}}$ could be ignored in~(\ref{equ:thermalmodel}) as well as the term $\sum_{j\in\mathcal{N}(i)}\frac{Z_j-Z_i}{R_{ij}}$ in~(\ref{equ:optapp}). Then following the same design procedure, we obtain another distributed control scheme (given by setting $R_{ij}=\infty$ in~(\ref{equ:controllerapp}a) and~(\ref{equ:controllerplus})) that does not need communication between neighboring zones, i.e., it requires less communication.
\end{rem}

\section{Numerical Investigations}\label{se:simulation}
In this section, we present two numerical examples for scenarios described in Sections~\ref{se:flowonly} and~\ref{se:general} respectively, using a house with four adjacent zones as illustrated in Figure~\ref{fig:1}. Only the heating case is presented in the following, while the cooling case is similar under the proposed control schemes.

The parameters of the simulations are obtained from~\cite{TahSMR11,DerFS11,LiWB16}: all $C_i=20$kJ$\text{/}^{\circ}$C, all $C_{fi}=35$kJ$\text{/}^{\circ}$C, all $C_{wi}=25$kJ$\text{/}^{\circ}$C, all $R_i=15^{\circ}$C/kW, all $R_{ij}=23^{\circ}$C/kW, all $R_{afi}=3^{\circ}$C/kW, all $R_{fwi}=5^{\circ}$C/kW, $c_w=4.186$kJ/kg$\text{/}^{\circ}$C, $[q_i^{max}]=[0.03, 0.04, 0.045, 0.035]$kg/s, $[T_s^{min},T_s^{max}]=[38,42]^{\circ}$C, $a=0.11\text{/}^{\circ}$C, $b=8.4$, all $r_i=0.5$p.u., all $k_{Z_i}=0.025$p.u., all $k_{Z_{fi}}=0.033$p.u., $k_{T_s}=0.05$p.u., all $k_{u_i}=\hat k_{eZ_{fi}}=k_{\zeta_i}=k_{\lambda_i}=k_{\mu_i^+}=k_{\mu_i^-}=k_{\nu^+}=k_{\nu^-}=1$p.u., all $k_{eu_i}=10$p.u., all $\hat k_{eu_i}=0.1$p.u., and all $k_{eZ_{fi}}=2$p.u. (p.u. means per unit). The outdoor temperature, indoor heat gains, and supply temperature in Scenario I are shown in Figure~\ref{fig:disturb}.

The simulation result of the first scenario is illustrated in Figures~\ref{fig:tempnoTs}-\ref{fig:ratenoTs}, in which we set $s=0$ before $15$h and $s=10$p.u. thereafter. The curves labelled with ``app'' indicate the case of using the decentralized controller given in Remark 2, i.e., a communication free scheme. It can be seen that the difference between these two cases is not large, i.e., less than $1.3^{\circ}$C in temperatures, indicating that the performance of the decentralized controller could be acceptable in practice. Before $15$h, since there is no consumption reduction purpose, i.e., $s=0$, the temperature trajectories under~(\ref{equ:controller}a)-(\ref{equ:controller}e),~(\ref{equ:controller}g)-(\ref{equ:controller}i) and~(\ref{equ:controllerqi})-(\ref{equ:controllerplus}) track their set points unless the corresponding water flow rate saturates (although not shown here, using a PID controller will result in the same temperature trajectories during this period). After $15$h, deviations from temperature set points appear due to the consideration of energy saving (while only using a PID controller can not reduce energy consumption unless forcing users to change their set points). On the other hand, we compare the flow rate response under the distributed controller, with and without extra dynamics. The result shown in Figure~\ref{fig:ratecompnoTs} demonstrates that after introducing those extra dynamics, the system performance has been improved that the oscillations are largely attenuated.

In the second scenario, $s=1$p.u. holds before $13$h and $s=5$p.u. thereafter. We can see that the difference between using the distributed controller~(\ref{equ:controllerapp}a)-(\ref{equ:controllerapp}e),~(\ref{equ:controllerapp}g)-(\ref{equ:controllerapp}k) and~(\ref{equ:controllerqi})-(\ref{equ:controllerplus}) and using its simplified version given in Remark 3 is also not large, as shown in Figures~\ref{fig:tempge}-\ref{fig:Tsge}. The temperature deviations with respect to their set points before increasing the weight coefficient $s$ are smaller than those thereafter since starting from $13$h, energy saving becomes more important while user comfort becomes less. All these two scenarios inspire us that tuning the weight coefficient $s$ (or equivalently $r_i$, as the optimal solution of~(\ref{equ:opt})/(\ref{equ:optapp}) depends on the ratio $r_i/s, i\in\mathcal{N}$) can balance user comfort and GHP system energy consumption -- there always exists a tradeoff between user comfort and energy saving.

\begin{figure}[!t]
\centering
\includegraphics[width=0.5\textwidth]{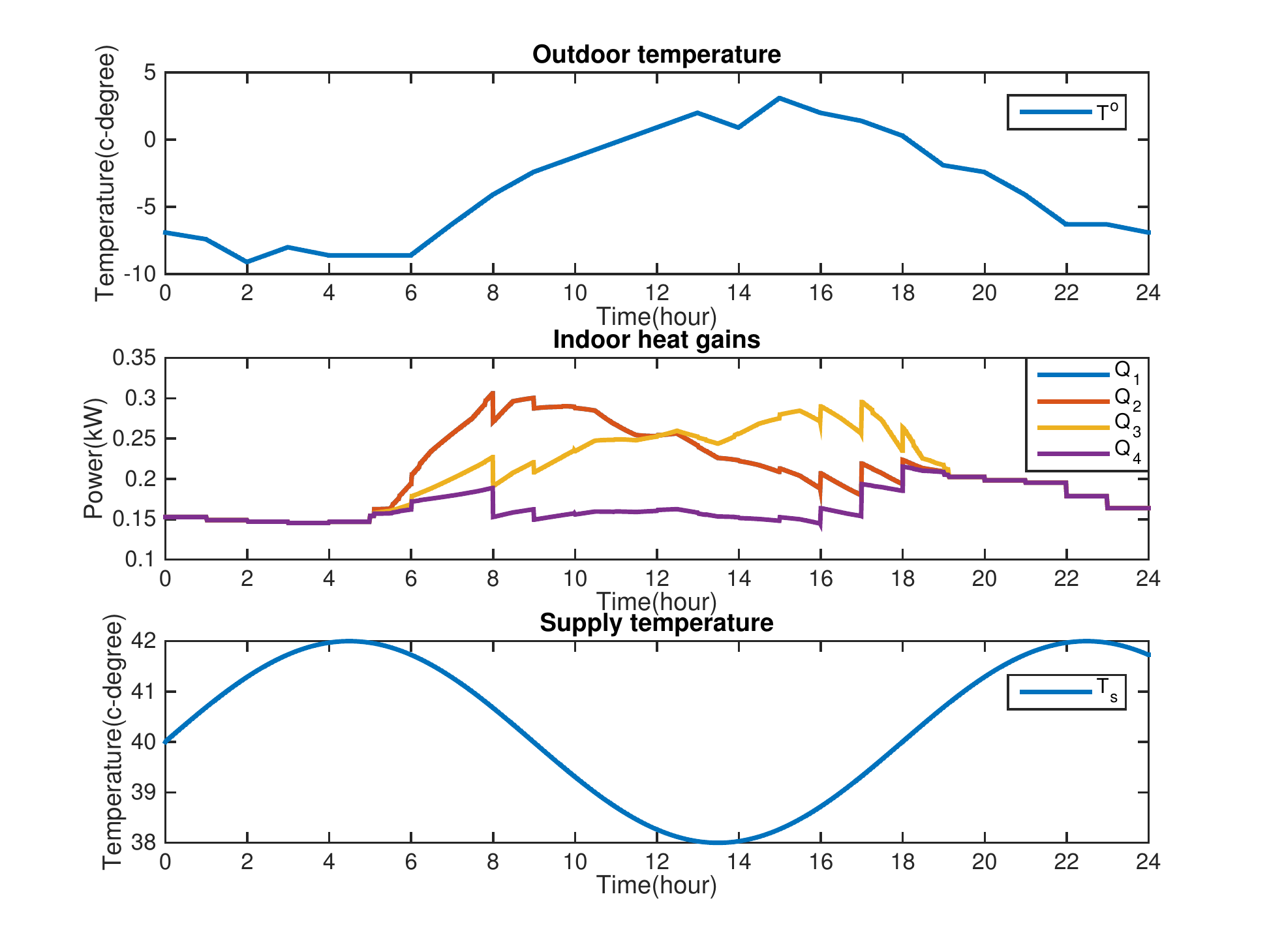}
\caption{Profiles of the exogenous inputs ($Q_1=Q_2$).}
\label{fig:disturb}
\end{figure}

\begin{figure}[!t]
\centering
\includegraphics[width=0.5\textwidth]{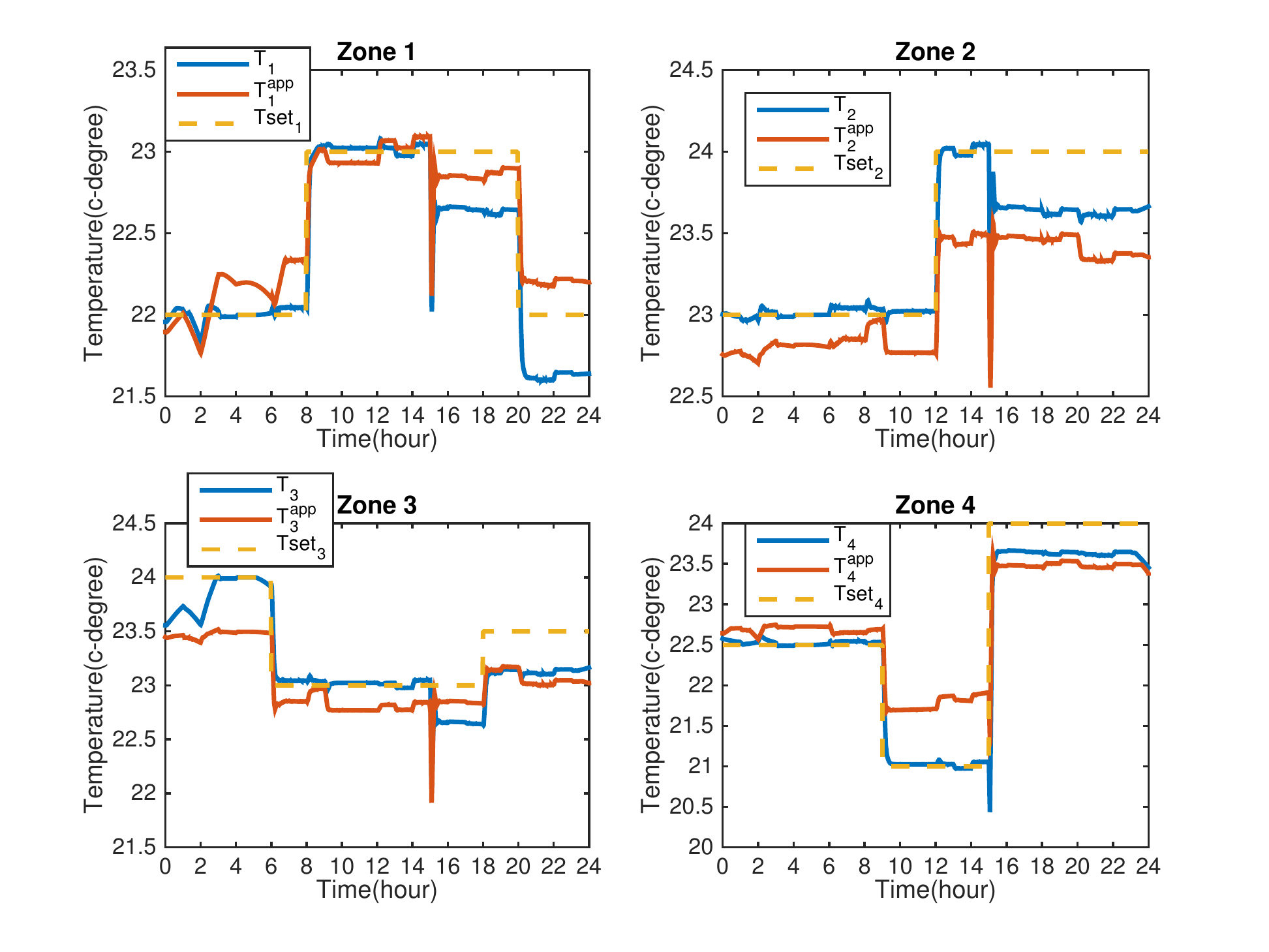}
\caption{Temperatures in Scenario I under~(\ref{equ:controller}a)-(\ref{equ:controller}e),~(\ref{equ:controller}g)-(\ref{equ:controller}i) and~(\ref{equ:controllerqi})-(\ref{equ:controllerplus}): curves labeled with ``app'' indicate the case of using the decentralized controller given in Remark 2.}
\label{fig:tempnoTs}
\end{figure}

\begin{figure}[!t]
\centering
\includegraphics[width=0.5\textwidth]{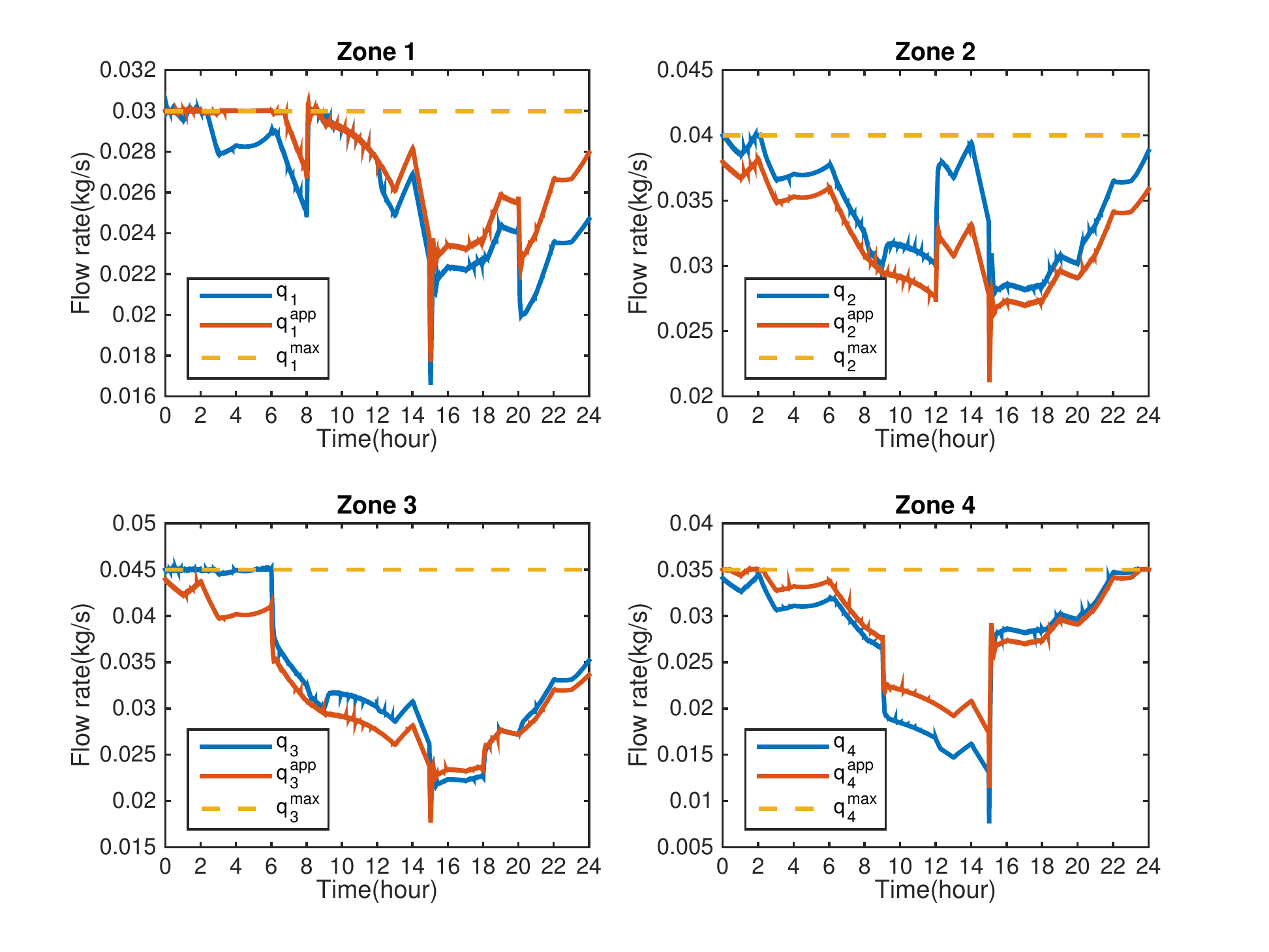}
\caption{Flow rates in Scenario I.}
\label{fig:ratenoTs}
\end{figure}

\begin{figure}[!t]
\centering
\includegraphics[width=0.5\textwidth]{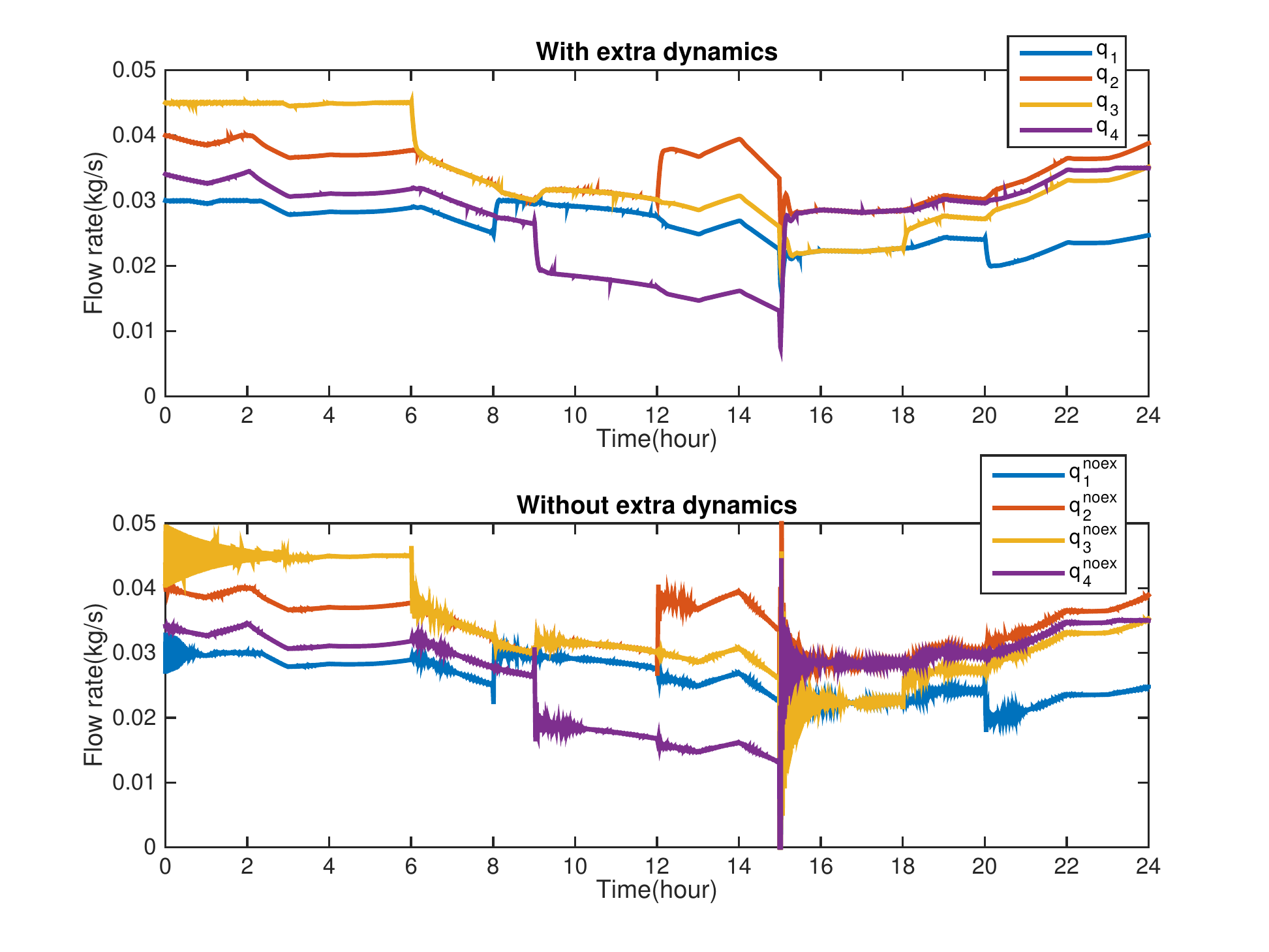}
\caption{Flow rates under~(\ref{equ:controller}a)-(\ref{equ:controller}e),~(\ref{equ:controller}g)-(\ref{equ:controller}i) and~(\ref{equ:controllerqi})-(\ref{equ:controllerplus}): with and without ($k_{eu_i}=k_{eZ_{fi}}=0$) extra dynamics.}
\label{fig:ratecompnoTs}
\end{figure}

\begin{figure}[!t]
\centering
\includegraphics[width=0.5\textwidth]{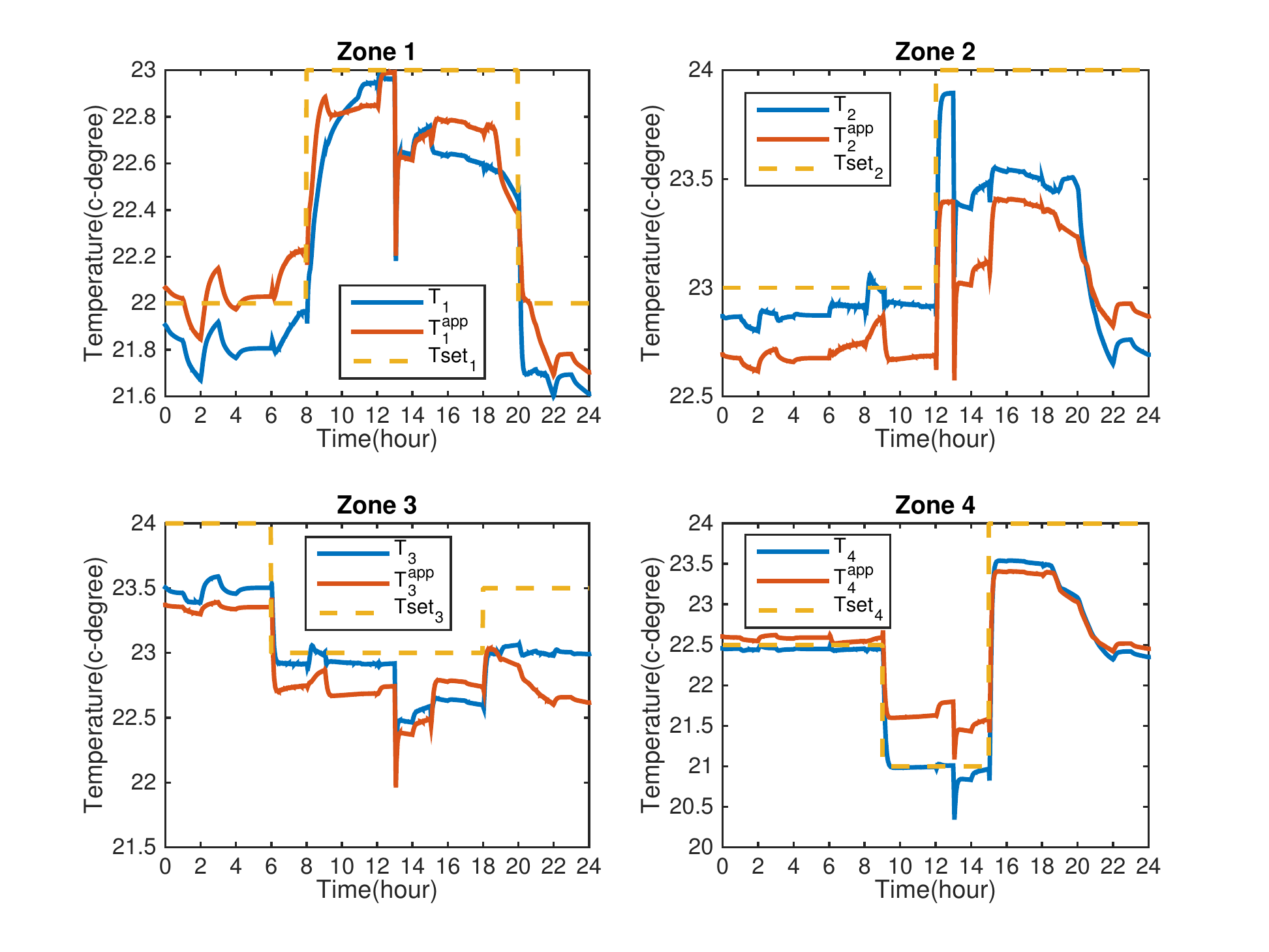}
\caption{Temperatures in Scenario II under~(\ref{equ:controllerapp}a)-(\ref{equ:controllerapp}e),~(\ref{equ:controllerapp}g)-(\ref{equ:controllerapp}k) and~(\ref{equ:controllerqi})-(\ref{equ:controllerplus}): curves labeled with ``app'' indicate the case of using the distributed controller given in Remark 3.}
\label{fig:tempge}
\end{figure}

\begin{figure}[!t]
\centering
\includegraphics[width=0.5\textwidth]{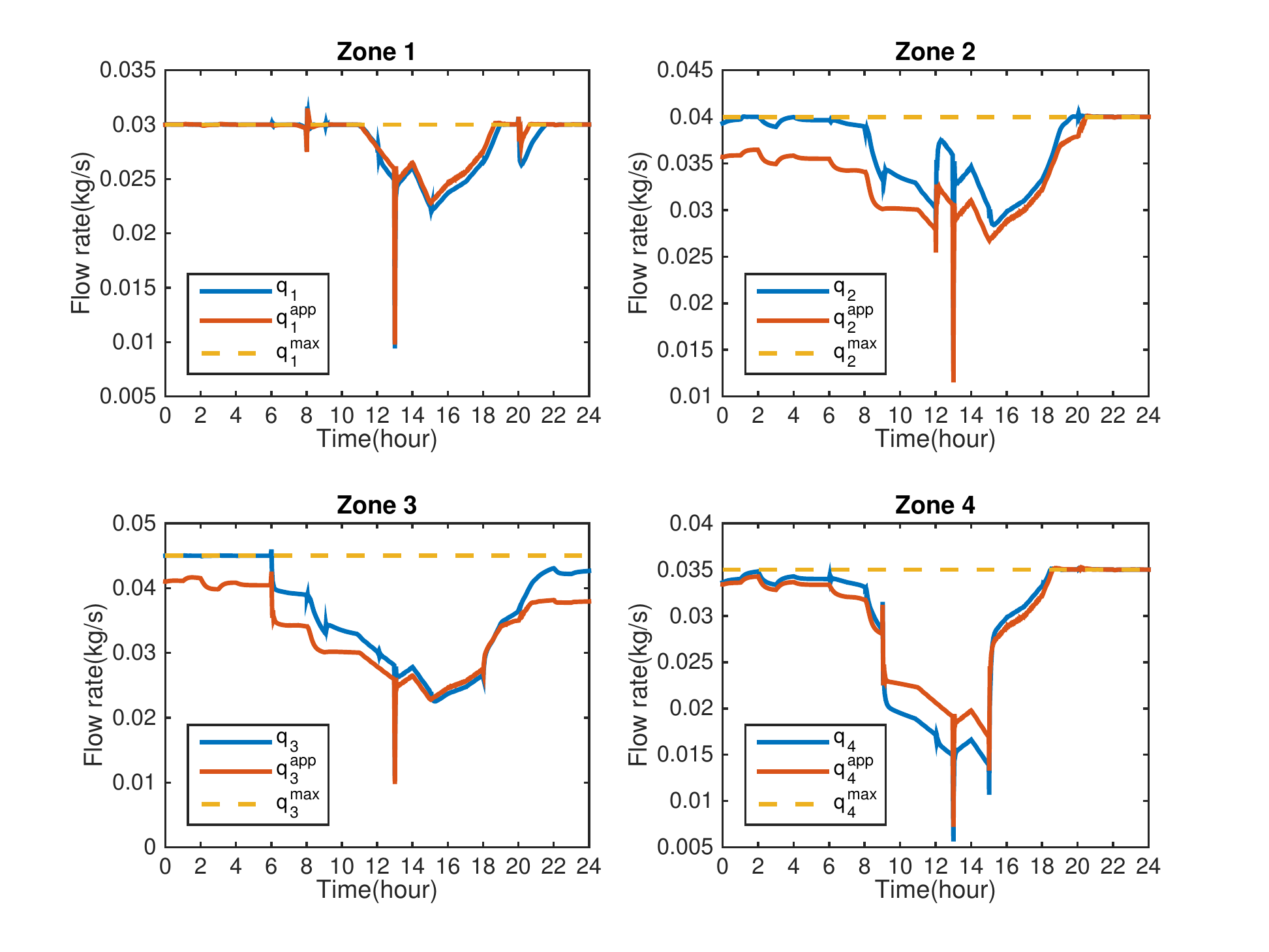}
\caption{Flow rates in Scenario II.}
\label{fig:ratege}
\end{figure}

\begin{figure}[!t]
\centering
\includegraphics[width=0.429\textwidth]{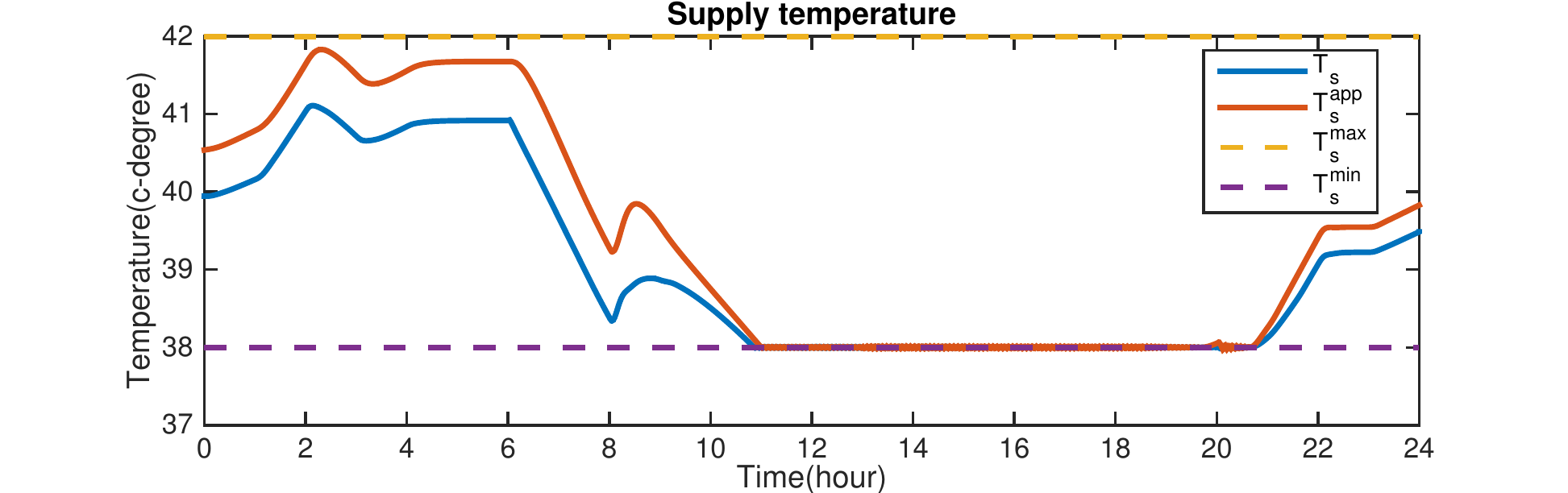}
\caption{Supply temperature in Scenario II.}
\label{fig:Tsge}
\end{figure}

\section{Conclusion and Future Work}\label{se:conclusion}
This paper presents distributed control frameworks on real-time temperature regulation via GHP combined with floor heating/cooling systems in energy efficient buildings. The controllers regulate water flow rates as well as the heat pump supply temperature, which balance user comfort and energy saving. Moreover, they can automatically adapt to changes of disturbances such as the outdoor temperature and indoor heat gains, without measuring or predicting those values. Also, the implementation of the controllers are simple.

Future work includes: developing distributed/decentralized control schemes to drive system~(\ref{equ:thermalmodel})-(\ref{equ:floormodel}) to the exact optimal solution of problem~(\ref{equ:opt}) (this may require an exact convex relaxation approach to solve~(\ref{equ:opt})); considering both floor and radiator heating/cooling (this may require model reduction for the dynamics of hydronic radiators); and investigating, e.g., the $H_2$ and $H_{\infty}$ performances of the controlled systems.

\bibliographystyle{IEEEtran}
\bibliography{zxbb}

\begin{thebibliography}{10}
\providecommand{\url}[1]{#1}
\csname url@samestyle\endcsname
\providecommand{\newblock}{\relax}
\providecommand{\bibinfo}[2]{#2}
\providecommand{\BIBentrySTDinterwordspacing}{\spaceskip=0pt\relax}
\providecommand{\BIBentryALTinterwordstretchfactor}{4}
\providecommand{\BIBentryALTinterwordspacing}{\spaceskip=\fontdimen2\font plus
\BIBentryALTinterwordstretchfactor\fontdimen3\font minus
  \fontdimen4\font\relax}
\providecommand{\BIBforeignlanguage}[2]{{%
\expandafter\ifx\csname l@#1\endcsname\relax
\typeout{** WARNING: IEEEtran.bst: No hyphenation pattern has been}%
\typeout{** loaded for the language `#1'. Using the pattern for}%
\typeout{** the default language instead.}%
\else
\language=\csname l@#1\endcsname
\fi
#2}}
\providecommand{\BIBdecl}{\relax}
\BIBdecl

\bibitem{UNEP09}
{UNEP Sustainable Buildings \& Climate Initiative}, ``Buildings and climate
  change: {S}ummary for decision-makers,'' \emph{Paris CEDEX 09, France:
  Sustainable United Nations}, 2009.

\bibitem{Yan13}
B.~Yan, ``A {B}ayesian approach for predicting building cooling and heating
  consumption and applications in fault detection,'' \emph{PhD dissertation,
  University of Pennsylvania}, 2013.

\bibitem{TahSR11ifac}
F.~Tahersima, J.~Stoustrup, and H.~Rasmussen, ``Optimal power consumption in a
  central heating system with geothermal heat pump,'' in \emph{Proc. 18th IFAC
  World Congress}, 2011, pp. 3102--3107.

\bibitem{YanPLT07}
Z.~Yang, G.~Pedersen, L.~Larsen, and H.~Thybo, ``Modeling and control of indoor
  climate using a heat pump based floor heating system,'' in \emph{Proc. 33rd
  Annual Conference of the IEEE Industrial Electronics Society}, 2007, pp.
  574--579.

\bibitem{HalPMJ12}
R.~Halvgaard, N.~K. Poulsen, H.~Madsen, and J.~B. J{\o}rgensen, ``Economic
  model predictive control for building climate control in a smart grid,'' in
  \emph{Proc. IEEE PES Innovative Smart Grid Technologies}, 2012.

\bibitem{TahSRM12}
F.~Tahersima, J.~Stoustrup, H.~Rasmussen, and S.~A. Meybodi, ``Economic {COP}
  optimization of a heat pump with hierarchical model predictive control,'' in
  \emph{Proc. of 51st IEEE Conference on Decision and Control}, 2012, pp.
  7583--7588.

\bibitem{ZhaP15}
X.~Zhang and A.~Papachristodoulou, ``A real-time control framework for smart
  power networks: {D}esign methodology and stability,'' \emph{Automatica},
  vol.~58, pp. 43--50, 2015.

\bibitem{ZhaKMP15}
X.~Zhang, R.~Kang, M.~McCulloch, and A.~Papachristodoulou, ``Real-time active
  and reactive power regulation in power systems with tap-changing transformers
  and controllable loads,'' \emph{Sustainable Energy, Grids and Networks},
  vol.~5, pp. 27--38, 2016.

\bibitem{ShiLXCG14}
W.~Shi, N.~Li, X.~Xie, C.-C. Chu, and R.~Gadh, ``Optimal residential demand
  response in distribution network,'' \emph{IEEE Journal on Selected Areas in
  Communications}, vol.~32, no.~7, pp. 1441--1450, 2014.

\bibitem{ShiLC16}
W.~Shi, N.~Li, C.~C. Chu, and R.~Gadh, ``Real-time energy management in
  microgrids,'' \emph{IEEE Transactions on Smart Grid}, 2016.

\bibitem{JadLM03}
A.~Jadbabaie, J.~Lin, and A.~S. Morse, ``Coordination of groups of mobile
  autonomous agents using nearest neighbor rules,'' \emph{IEEE Transactions on
  Automatic Control}, vol.~48, no.~6, pp. 988--1001, 2003.

\bibitem{Wan10}
F.-Y. Wang, ``Parallel control and management for intelligent transportation
  systems: {C}oncepts, architectures, and applications,'' \emph{IEEE
  Transactions on Intelligent Transportation Systems}, vol.~11, no.~3, pp.
  630--638, 2010.

\bibitem{ChaMA10}
V.~Chandan, S.~Mishra, and A.~G. Alleyne, ``Predictive control of complex
  hydronic systems,'' in \emph{Proc. of 2010 American Control Conference},
  2010, pp. 5112--5117.

\bibitem{Tah12}
F.~Tahersima, ``An integrated control system for heating and indoor climate
  applications,'' \emph{PhD dissertation, Aalborg University}, 2012.

\bibitem{TahSMR11}
F.~Tahersima, J.~Stoustrup, S.~A. Meybodi, and H.~Rasmussen, ``Contribution of
  domestic heating systems to smart grid control,'' in \emph{Proc. of 50th IEEE
  Conference on Decision and Control and European Control Conference}, 2011,
  pp. 3677--3681.

\bibitem{LinMB12}
Y.~Lin, T.~Middelkoop, and P.~Barooah, ``Issues in identification of
  control-oriented thermal models of zones in multi-zone buildings,'' in
  \emph{Proc. of 51st IEEE Conference on Decision and Control}, 2012, pp.
  6932--6937.

\bibitem{BoyV04}
S.~Boyd and L.~Vandenberghe, \emph{Convex Optimization}.\hskip 1em plus 0.5em
  minus 0.4em\relax Cambridge University Press, 2004.

\bibitem{ZhaSL17}
X.~Zhang, W.~Shi, B.~Yan, A.~Malkawi, and N.~Li, ``Decentralized and
  distributed temperature control via {HVAC} systems in energy efficient
  buildings,'' \emph{arXiv:1702.03308 [cs.SY]}, 2017.

\bibitem{HaoLKS15}
H.~Hao, J.~Lian, K.~Kalsi, and J.~Stoustrup, ``Distributed flexibility
  characterization and resource allocation for multi-zone commercial buildings
  in the smart grid,'' in \emph{Proc. of 54th IEEE Conference on Decision and
  Control}, 2015, pp. 3161--3168.

\bibitem{ZhaP15tac}
X.~Zhang and A.~Papachristodoulou, ``Improving the performance of network
  congestion control algorithms,'' \emph{IEEE Transactions on Automatic
  Control}, vol.~60, no.~2, pp. 522--527, 2015.

\bibitem{ZhaLP15cdc}
X.~Zhang, A.~Papachristodoulou, and N.~Li, ``Distributed optimal steady-state
  control using reverse- and forward-engineering,'' in \emph{Proc. of 54th IEEE
  Conference on Decision and Control}, 2015, pp. 5257--5264.

\bibitem{FeiP10}
D.~Feijer and F.~Paganini, ``Stability of primal-dual gradient dynamics and
  applications to network optimization,'' \emph{Automatica}, vol.~46, no.~12,
  pp. 1974--1981, 2010.

\bibitem{CheMC16}
A.~Cherukuri, E.~Mallada, and J.~Cort\'{e}s, ``Asymptotic convergence of
  constrained primal-dual dynamics,'' \emph{System and Control Letters},
  vol.~87, pp. 10--15, 2016.

\bibitem{DerFS11}
M.~Deru, K.~Field, D.~Studer, K.~Benne, B.~Griffith, P.~Torcellini, B.~Liu,
  M.~Halverson, D.~Winiarski, M.~Rosenberg, M.~Yazdanian, J.~Huang, and
  D.~Crawley, ``Department of energy commercial reference building models of
  the national building stock,'' in \emph{Technical Report NREL/TP-5500-46861},
  2011.

\bibitem{LiWB16}
X.~Li, J.~Wen, and E.~W. Bai, ``Developing a whole building cooling energy
  forecasting model for on-line operation optimization using proactive system
  identification,'' \emph{Applied Energy}, vol. 164, pp. 69--88, 2016.

\end{thebibliography}
\bibliographystyle{choosenstyle}

\end{spacing}

\end{document}